\documentclass[conference]{IEEEtran}
\IEEEoverridecommandlockouts

\usepackage[numbers,sort&compress]{natbib}
\usepackage{amsmath,amssymb,amsfonts,amsthm}
\usepackage{mathrsfs}
\usepackage{mathtools}
\usepackage[bbgreekl]{mathbbol}
\usepackage{booktabs}
\usepackage[stable]{footmisc}
\usepackage{algorithmic}
\usepackage{enumitem}
\usepackage{graphicx}
\usepackage{textcomp}
\usepackage{tikz}
\usetikzlibrary{arrows}
\usepackage{standalone}
\usepackage{pgfplots}
\newlength\figH
\newlength\figW
\usepackage{xcolor, color, colortbl}
\usepackage[T1]{fontenc}
\usepackage[utf8]{inputenc}
\usepackage{xpatch}
\usepackage{fancyhdr}

\def\BibTeX{{\rm B\kern-.05em{\sc i\kern-.025em b}\kern-.08em
    T\kern-.1667em\lower.7ex\hbox{E}\kern-.125emX}}

\usepackage{titlesec}

\newcommand*{\C}{\mathbb{C}}
\newcommand*{\R}{\mathbb{R}}
\newcommand*{\N}{\mathbb{N}}

\DeclareMathOperator{\diff}{d}
\DeclareMathOperator{\diag}{diag}

\newcommand{\ddt}{\tfrac{\diff}{\diff \!t}}
\definecolor{light-gray}{gray}{0.96}

\newtheorem{theorem}{Theorem}
\newtheorem{lemma}{Lemma}

\newtheorem{condition}{Condition}

\DeclareSymbolFontAlphabet{\mathbb}{AMSb}
\DeclareSymbolFontAlphabet{\mathbbl}{bbold}
\makeatletter
\xpatchcmd{\@thm}{\thm@headpunct{.}}{\thm@headpunct{}}{}{}
\makeatother

\makeatletter
\let\MYcaption\@makecaption
\makeatother
\usepackage[font=footnotesize]{subcaption}
\makeatletter
\let\@makecaption\MYcaption
\makeatother

\begin{document}
\thispagestyle{plain}
\pagestyle{plain}

\title{{\fontsize{16}{19}\selectfont \textbf{Frequency Domain Stability Conditions for Hybrid AC/DC Systems}}}

\author{\IEEEauthorblockN{{\fontsize{12}{14}\selectfont Dahlia Saba and Dominic Gro\ss}}
\IEEEauthorblockA{\textit{\fontsize{10}{12}\selectfont Department of Electrical and Computer Engineering} \\
\textit{\fontsize{10}{12}\selectfont University of Wisconsin--Madison}\\
\fontsize{10}{12}\selectfont Madison, WI, USA \\
\fontsize{10}{12}\selectfont \{dsaba, dominic.gross\}@wisc.edu
 \thanks{This material is based upon work supported by the National Science Foundation under Grant No. 2143188.}}}

\maketitle

\begin{abstract}
    In this article, we investigate small-signal frequency and DC voltage stability of hybrid AC/DC power systems that combine AC and DC transmission, conventional machine-based generation, and converter-interfaced generation. The main contributions of this work are a compact frequency domain representation of hybrid AC/DC systems and associated stability conditions that can be divided into conditions on the individual bus dynamics and conditions on each DC network. The bus-level conditions apply to a wide range of technologies (e.g., synchronous generators, synchronous condensers, grid-forming renewables and energy storage). Moreover, the system-level conditions establish that hybrid AC/DC systems combining a wide range of devices are stable independently of the network topology provided that the frequency response of converters on each DC network is sufficiently coherent relative to the network coupling strength. Additionally, we develop and validate a novel reduced-order damper winding model for multi-machine systems.
\end{abstract}

\section{Introduction}
In an effort to decarbonize power systems and the wider economy, bulk power systems are rapidly transitioning from conventional fuel-based generation interfaced by synchronous machines to emerging technologies interfaced by power electronic converters such as renewable generation~\cite{KJZ+2017}, energy storage, and high voltage DC (HVDC) transmission~\cite{WLB2015}. However,  replacing conventional machine-interfaced generation and AC transmission with converter-interfaced resources, storage, and transmission results in significantly different power system dynamics and challenges standard operating and analysis paradigms~\cite{WLB2015,KJZ+2017,MDH2018}. For example, a widely recognized concern is the loss of synchronous generators that stabilize power systems through their physical properties (e.g., rotational inertia) and controls (e.g., speed governor). In contrast, today's converter-interfaced renewables are often controlled to maximize energy yield and can jeopardize system reliability.

The vast majority of power converters connected to bulk power systems today use so-called grid-following (GFL) control that assumes a stable AC voltage waveform (i.e., frequency and magnitude) at their point of interconnection. While GFL converters can provide grid support services, dynamic stability of the bulk power system can rapidly deteriorate as the share of GFL resources increases~\cite{MBP+2019}. In contrast, grid-forming (GFM) converters are controlled as stable and self-synchronizing voltage sources that can mimic the stabilizing response of synchronous generators. Grid-forming controls commonly discussed in the literature include droop control~\cite{CDA1993}, virtual synchronous machine (VSM) control~\cite{DARCO2015180}, (dispatchable) virtual oscillator control (dVOC)~\cite{JDH2014,dVOC}, machine emulation~\cite{7236454} and dual-port GFM control~\cite{universaldualport}. However, droop control, VSM, and dVOC may destabilize the system if, e.g., the resource interfaced by the converter reaches its
power generation limits or if the topology or loading of HVDC networks changes~\cite{singleportunstable}. Overall, stability analysis should consider heterogeneous controls, converter DC bus dynamics, including power generation resources and DC transmission, and legacy generation (i.e., synchronous generators).

Stability analysis methods in the literature can be broadly categorized into large-signal~\cite{foundations_chiang_1987, transient_fouad_1988, synchronization_dorfler_2012, SCHIFFER20142457, timescale_subotic_2021,large-signal_pal_2023} and small-signal~\cite{input-admittance_harnefors_2007,PM2019,universaldualport,huang_gain_2024} analysis methods. For instance, energy-functions have been used to establish large-signal stability of multi-machine systems~\cite{foundations_chiang_1987, transient_fouad_1988,large-signal_pal_2023} and to derive analytic AC stability conditions for a single GFL converter~\cite{large-signal_pal_2023}. However, at present, these methods are not scalable to large-scale networks of heterogeneous devices.

Small-signal analysis using impedance models~\cite{input-admittance_harnefors_2007, sun_impedance_2011,wang_small-signal_2024,zhang_impedance_2021,zhang_harmonic_2022,huang_gain_2024,arevalo-soler_small-signal_2025} and eigenvalue sensitivity analysis~\cite{HAROLARRODE2020105746,arevalo-soler_small-signal_2025} have proven useful to analyze harmonic stability and interactions of various converter control loops. In particular, by ensuring that the converter admittance is passive in ranges where the network may be vulnerable to resonances, converters can help stabilize the system~\cite{harnefors_passivity-based_2016}. These methods generally require detailed knowledge of the network and converter topology and numerical analysis may not be tractable for large systems. A notable exception are the decentralized impedance-based conditions for AC systems explored in~\cite{huang_gain_2024}. Impedance analysis has also been used extensively to study the impact of inner control loops~\cite{input-admittance_harnefors_2007, sun_impedance_2011, harnefors_passivity-based_2016, huang_gain_2024, wang_small-signal_2024,zhang_impedance_2021,zhang_harmonic_2022, HAROLARRODE2020105746}.

Conceptually, the hybrid AC/DC systems considered in this work can be represented and analyzed using hybrid AC/DC impedance models~\cite{zhang_impedance_2021,zhang_harmonic_2022,wang_small-signal_2024}. However, the literature on hybrid AC/DC impedance analysis typically only considers simplified network topologies such as (i) a single interlinking converter  connecting an infinite bus (i.e., AC system) to a DC impedance (i.e., DC system)~\cite{wang_small-signal_2024}, or (ii) two AC networks (represented as infinite bus) connected via a point-to-point HVDC line~\cite{zhang_impedance_2021}. In~\cite{zhang_harmonic_2022}, a multi-terminal HVDC system is considered; however, this work neglects the dynamics of AC generation. Moreover,~\cite{wang_small-signal_2024,zhang_impedance_2021,zhang_harmonic_2022} restrict the converter to GFL control.

In contrast, the focus of this work are scalable analytical stability conditions that do not restrict the network topology or converter control a-priori. Justified by standard timescale separation arguments (see~\cite{timescale_subotic_2021} and references therein), we focus on the frequency and DC voltage dynamics of hybrid AC/DC systems that arise from interactions of converter outer control loops, resource dynamics, and synchronous generators through AC and DC networks. Conceptually, this work builds on analytical methods for small-signal stability analysis that are often are restricted to AC networks with homogeneous bus dynamics~\cite{paganini_global_2020} or GFM converters with identical controls. A notable exception are the conditions in \cite{PM2019} that characterize a general class of bus (e.g., generators and converters) transfer functions to ensure frequency stability of AC networks independently of the topology. However, \cite{PM2019} requires that individual bus transfer functions are $\mathscr{H}_\infty$ stable and thereby excludes, e.g., synchronous condensers, GFM STATCOMS, and renewables operating at their maximum power point (MPP). Moreover, the results do not extend to hybrid AC/DC power systems. In contrast, \cite{universaldualport} presents analytical stability conditions for hybrid AC/DC power systems containing converters and synchronous machines that allow for synchronous condensers and renewables at the MPP. However, the results in \cite{universaldualport} only apply to converters using dual-port GFM control and restrict the set of allowable AC network topologies.

The main contribution of this work is a compact frequency domain representation of hybrid AC/DC systems and associated stability conditions that can incorporate common technologies such as synchronous generators, synchronous condensers, HVDC transmission, and common GFM controls. We first develop an abstract small-signal modeling framework that models a wide range of emerging and conventional technologies. Next, we develop and validate a novel reduced-order model of damper windings that captures their synchronizing effect in multi-machine systems. This model is crucial to incorporate synchronous condensers into our stability analysis framework and is validated using EMT simulations.

Leveraging this modeling framework we derive analytical stability conditions that can be divided into bus-level and system-level conditions. The bus-level conditions allow for marginally stable bus dynamics, thereby allowing us to consider synchronous condensers and renewables operating at their MPP. Moreover, the system-level conditions only require knowledge of the strongest and weakest coupling in the grid but do not require knowledge of the precise system topology. Thus, our conditions are both inherently scalable and can capture a wide range of generation technologies that cannot be analyzed using existing frameworks.

This manuscript is structured as follows. We briefly review small-signal models of common power generation, conversion, and transmission technologies in Sec.~\ref{sec:modeling}. Our reduced-order damper winding model is introduced and validated in Sec.~\ref{sec:dw}. Next, Sec.~\ref{sec:inter} introduces the small-signal model of the overall hybrid AC/DC system and Sec.~\ref{sec:stability} develops our analytical stability conditions. Finally, Sec.~\ref{sec:applicationexample} discusses the application of the analytical stability conditions to common applications and Sec.~\ref{sec:conclusion} provides conclusions and directions for future work.

\section{Review of Power System Models}\label{sec:modeling}
This section reviews common reduced-order models of power systems, including models of transmission networks, of power converters and converter-interfaced generation, and of synchronous machines and mechanical power generation. 

\subsection{AC and DC Transmission Network Model}\label{subsec:acanddctransmission}
The transmission system is modeled as a Kron-reduced graph~\cite{kron}, where the $n \in \N$  nodes represent buses with converters or synchronous machines and the edges represent AC and DC lines. Moreover, $\mathcal{N}_{\text{ac}} \subseteq \{1, \dots, n\}$ and $\mathcal{N}_{\text{dc}} \subseteq \{1, \dots, n\}$ represent the sets of nodes connected to AC lines and DC lines, respectively. Finally, $\mathcal{E}_{\text{ac}} \subseteq \mathcal{N}_{\text{ac}} \times \mathcal{N}_{\text{ac}}$ and $\mathcal{E}_{\text{dc}} \subseteq \mathcal{N}_{\text{dc}} \times \mathcal{N}_{\text{dc}}$ represent the sets of AC and DC lines.  

It is assumed that the AC voltage magnitude for every bus is constant at $1$~p.u., all AC lines are lossless, and reactive power flows do not impact the AC phase angles. Linearizing the AC powerflow equations around the no load operating point, the power $p^{\text{ac}}_{\text{net,}i}$ injected at bus $i \in \mathcal{N}_{\text{ac}}$ into the AC network is
\begin{align*}
    p^{\text{ac}}_{\text{net,}i} = \sum\nolimits_{(i,j)\in \mathcal{E}_{\text{ac}}} b^{\text{ac}}_{ij}(\theta_i - \theta_j),
\end{align*}
where $b^{\text{ac}}_{ij} \in \R$ represents the susceptance of the AC line $(i,j)$ and $\theta_i  \in \R$ represents the AC voltage phase angle of bus $i \in \mathcal{N}_{\text{ac}}$. This equation can be written in vector form as 
\begin{align}
    \tilde{P}^{\text{ac}}_{\text{net}} = \tilde{L}_{\text{ac}}\theta, 
\end{align}
where $\tilde{P}^{\text{ac}}_{\text{net}} \in \R^{|\mathcal{N}_{\text{ac}}|}$ is the vector of power injections into the AC network from each node, $\theta \in \R^{|\mathcal{N}_{\text{ac}}|}$ is the vector of AC voltage phase angles, and $\tilde{L}_{\text{ac}} \in \R^{|\mathcal{N}_{\text{ac}}| \times |\mathcal{N}_{\text{ac}}|}$ is a Laplacian matrix defined in Appendix~\ref{app:lap}.

Similarly, linearizing the powerflow in the DC network at the nominal DC voltage results in the power  
\begin{align*}
    p^{\text{dc}}_{\text{net,}i} = \sum\nolimits_{(i,j)\in \mathcal{E}_{\text{dc}}} g^{\text{dc}}_{ij}(v_i - v_j),
\end{align*}
injected into the DC grid at bus $i \in \mathcal{N}_{\text{dc}}$. Here, $g^{\text{dc}}_{ij} \in \R_{\geq 0}$ denotes the conductance of DC line $(i,j)$ and $v_i \in \R$ represents the DC voltage deviation at bus $i \in \mathcal{N}_{\text{dc}}$. This equation can be written in vector form as 
\begin{align}
    \tilde{P}^{\text{dc}}_{\text{net}} = \tilde{L}_{\text{dc}}v, 
\end{align}
where $\tilde{P}^{\text{dc}}_{\text{net}} \in \R^{|\mathcal{N}_{\text{dc}}|}$ is the vector of power injections into the AC network from each node, $v \in \R^{|\mathcal{N}_{\text{dc}}|}$ is the vector of DC voltages, and $\tilde{L}_{\text{dc}} \in \R^{|\mathcal{N}_{\text{dc}}| \times |\mathcal{N}_{\text{dc}}|}$ is the Laplacian matrix of the DC network defined in Appendix~\ref{app:lap}.

\subsection{DC Generation, Power Converters, and Controls} 
DC generation and converters will be modeled as a set of transfer functions from the total (i.e., AC and DC) converter power injection to the AC frequency and DC voltage at the converter bus (Fig.~\ref{fig:dc_blockdiag}). Any DC generation is assumed to located at a bus with a converter.
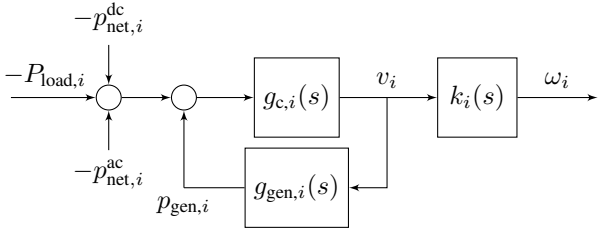
\begin{figure}[t!]
    \centering
    \tikzset{
block/.style = {draw, fill=white, rectangle, minimum height=3em, minimum width=3em},
tmp/.style  = {coordinate}, 
sum/.style= {draw, fill=white, circle, node distance=1cm},
input/.style = {coordinate},
output/.style= {coordinate,node distance=1cm},
pinstyle/.style = {pin edge={to-,thin,black}
}
}

\begin{tikzpicture}[auto, node distance=2cm,>=latex']
    \node [input, name=rinput] (rinput) {};
    \node [sum, right of=rinput, node distance=1.3cm] (sum1) {};
    \node [input, below of=sum1, node distance=0.7cm] (pac_net) {};
    \node [input, above of=sum1, node distance=0.7cm] (pdc_net) {};
    \node [sum, right of =sum1] (sum2) {};
    
    \node [block, right of=sum2, node distance=1.5cm] (controller) {$g_{\text{c,}i}(s)$};
          
   \node [output,right of=controller, node distance=1.2cm] (vout) {};    

    \node [block, below of=controller,node distance=1.2cm] (DCgen) {$g_{\text{gen,}i}(s)$};    
        \node [output, left of=DCgen, node distance=1.5cm] (pgen) {};
    
    \node [block, right of=vout, node distance=1.2cm] (k) {$k_i(s)$};
    \node [output, right of=k, node distance = 1.6cm] (output){};

    \draw [->] (rinput) -- node[pos=0.4]{$-P_{\text{load,}i}$} (sum1);
    \draw [->] (pac_net) -- node[below, pos=0]{$-p_{\text{net},i}^{\text{ac}}$} (sum1);
    \draw [->] (pdc_net) -- node[above, pos=0]{$-p_{\text{net},i}^{\text{dc}}$} (sum1);
    \draw [->] (sum1) -- (sum2);
    \draw [->] (sum2) --  (controller);
    \draw [-] (controller) --  node[pos=1]{$v_i$}  (vout);
    \draw [->] (vout) -- (k);
    \draw [->] (k) -- node{$\omega_i$} (output);
    
    \draw [-] (DCgen) -- node[pos=1]{$p_{\text{gen,}i}$} (pgen);
    \draw [->] (pgen) -- node[pos=0.9]{} (sum2);
    
         \draw [->] (vout) |-  (DCgen);    

    \end{tikzpicture}
    \caption{Block diagram of the small-signal model of DC generation with a converter at a bus $i$.}\label{fig:dc_blockdiag}
\end{figure}

\subsubsection{DC Generation}
The dynamic response of the DC generation at a bus $i \in \mathcal{N}_{\text{dc}}$ can be modeled as a transfer function $g_{\text{gen,}i}(s)$ from the DC voltage $v_i \in \R$ to the change in power injection $p_{\text{gen},i} \in \R$ of the DC source. For example, for a solar PV array, $g_{\text{gen}}(s)$ is used to linearly approximate the sensitivity of the PV power-voltage curve (see Fig.~\ref{fig:pv}).
\begin{figure}[b!]
    \centering
    \includegraphics{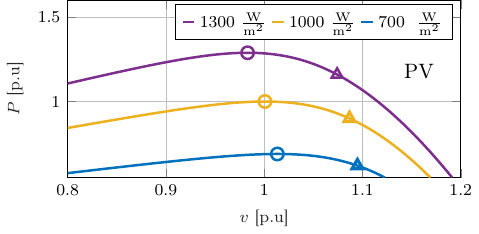}
    \caption{PV system output power $P$ as a function of the DC voltage $v$ under different insolation.}\label{fig:pv}
\end{figure}
Thus, a PV system operating at a voltage above its MPP voltage will have $g_{\text{gen,}i}(s) = -k_{\text{pv,}i}$ for some $k_{\text{pv,}i} \in \R_{>0}$, while a PV system operating at its MPP will have $g_{\text{gen,}i}(s) = 0$. Similarly, a battery providing proportional DC voltage control can be modeled as $g_{\text{gen,}i}(s) = -k_{\text{batt,}i}$ with $k_{\text{batt,}i} \in \R_{>0}$.

\subsubsection{Voltage Source Converters\footnote{The model presented in this paper only applies to voltage source converter topologies whose DC-link capacitor is directly connected to DC grids. This framework can be extended to cover, e.g., modular multilevel converters by introducing an additional transfer function mapping internal energy to DC voltage~\cite[Fig. 3]{MMC_ludois_2014}\cite{sanchez_mmc_energy}.} and Controls}\label{sec:conv_model}
The dynamics of a converter at bus $i \in \N$ is represented as two transfer functions. The transfer function $g_{\text{c,}i}(s)$ models the physics of the converter, i.e., the mapping from net power injected into bus $i \in \N$ to the DC voltage deviation $v_i \in \R$. For a two-level VSC with DC link capacitance $C_{\text{dc}}$, $g_{\text{c,}i}(s) = 1/C_{\text{dc}}s$. Using standard timescale separation arguments, we neglect the dynamics of the converter's inner control loops~\cite{timescale_subotic_2021}. Then, the outer controls of the converter are modeled by a transfer function $k_i(s)$ mapping from the DC voltage deviation $v_i  \in \R$ to the AC frequency deviation $\omega_i = \ddt \theta_i \in \R$. See Section~\ref{sec:applicationexample} for further details. 


\subsection{Synchronous Machines and Mechanical Power Generation}\label{sec:model_sm} 
\subsubsection{Synchronous Machines}
The dynamics of a synchronous machine $k \in \mathcal{N}_{\text{ac}}$ are modeled via the swing equation
\begin{subequations}\label{eq:swing_eq}
    \begin{align}
        \ddt \theta_{k} &= \omega_{\text{r,}k}, \\
        J_k\omega_0 \ddt \omega_{\text{r,}k} &= P_{\text{m,}k} - P_{\text{load,}k} - P_{\text{net,}k}, 
    \end{align}
\end{subequations}
where $J_k \in \R_{>0}$ is the machine's rotational inertia, $\omega_0 \in \R$ is the nominal frequency, $\omega_{\text{r,}k} \in \R$ is the deviation of the electrical frequency of the machine from the nominal frequency $\omega_0$, $P_{\text{m,}k} \in \R$ is the mechanical power generated by mechanical power generation (e.g., steam turbine), $P_{\text{net,}k} \in \R$ is the power injected into the AC network, and $P_{\text{load,}k} \in \R$ is the power consumption of loads at the machine bus. These dynamics can be represented by a transfer function
\begin{align*}
g_{\text{SM,}k}(s) = \frac{1}{J_k \omega_0s}
\end{align*}
from the net power injected into node $k \in \N$ to the machine's frequency deviation.

\subsubsection{Mechanical Power Generation}
The dynamic response of the mechanical power generation at a bus $k$ can be described as a transfer function $g_{\text{gen}}(s)$ from the machine's frequency deviation $\omega_{\text{r,}k}$ to the change in mechanical power generation $P_{\text{m},k}$. The dynamics of a turbine-governor system can be described as
\begin{align}
    \tau_k \ddt P_{\text{m,}k} = -P_{\text{m,}k} - k_{\text{g,}k}(\omega_{\text{r,}k} - \omega_0) \label{eq:turb_gov},
\end{align}
where $\tau_k \in \R_{\geq 0}$ is the turbine time constant, and $k_{\text{g,}k} \in \R_{\geq 0}$ is the governor gain. Therefore, for a machine with a turbine-governor, we obtain \[g_{\text{gen,}k}(s) = \frac{-k_{\text{g,}k}}{1 + \tau_k s}.\] In contrast, for a synchronous condenser we obtain $g_{\text{gen,}k}(s) = 0$. Next, we consider the aerodynamics of a wind turbine rotor.
%
%
%
If the wind turbine is operating at or above the speed and pitch angle corresponding to the MPP and controlling its pitch angle in proportion to rotor speed, we obtain
\begin{align*}
    g_{\text{gen,}k}(s) = -k_{\omega,k} - \frac{k_{\beta,k}k_{p,k}}{\tau_{\beta,k} s + 1},
\end{align*}
where $k_{\omega,k} \in \mathbb{R}_{\geq 0}$ is the sensitivity of the wind turbine's mechanical power output to changes in rotor frequency, $k_{\beta,k}  \in \mathbb{R}_{\geq 0}$ is the sensitivity of its mechanical power output to changes in pitch angle, $\tau_{\beta,k}$ is the time constant of the pitch motor, and $k_{p,k}$ is the control gain of the pitch control~\cite{universaldualport}. Notably, $k_{\beta,k}=0$ when operating at the optimal pitch angle and $k_{\omega,k}=0$ when operating at the optimal rotor speed, so $g_{\text{gen,}k}(s)=0$ when operating at the MPP.

\section{Reduced-order Modeling of Damper Windings}\label{sec:dw}
The model for synchronous machines presented in the previous section ignores the effects of damper windings on their frequency dynamics. This section presents a novel damper winding model which accounts for the effect of damper windings on the frequency dynamics of machines on a network.
To this end, we model a synchronous machine as two distinct buses in the network: (i) a bus representing the machine's back EMF voltage with electrical angle $\theta_{\text{r}}$ and frequency $\omega_{\text{r}}$ corresponding to the rotor speed, and (ii) a bus representing the interconnection point between the machine and the transmission grid  with AC voltage phase angle $\theta_{\text{s}}$ and frequency $\omega_{\text{s}}$ corresponding to the voltage of the stator (Fig.~\ref{fig:rs_buses}). 
\begin{figure}[b!]
    \centering
    \includegraphics[width=0.25\textwidth]{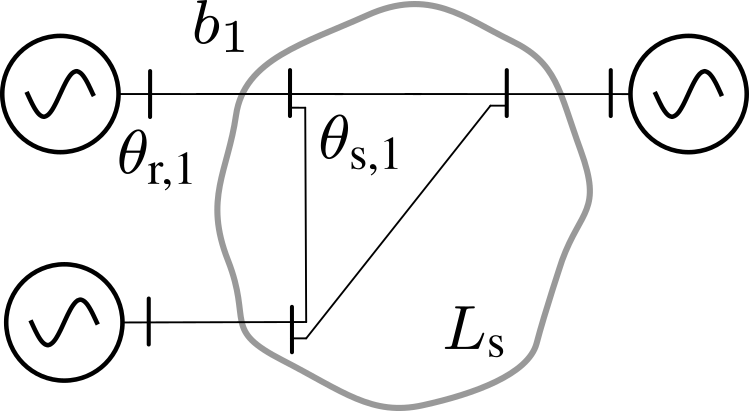}
    \caption{Diagram of the multi-machine system with damper windings. For the machine in the upper left, $\theta_{\text{r,1}}$ represents the AC voltage phase angle associated with the rotor, $\theta_{\text{s,1}}$ represents the AC voltage phase angle of the machine at the point of interconnection with the grid, and $b_1$ represents the susceptance of the stator windings. The (loopy) Laplacian $L_{\text{s}}$ represents the electrical network connecting the stator buses.}\label{fig:rs_buses}
\end{figure}
The susceptance of the connection between the two buses represents the susceptance of the machine's stator windings and is denoted as $b_k \in \R_{>0}$.

\subsection{Review of reduced-order damper winding model}
When the difference between the back EMF frequency $\omega_{\text{r}}$ and grid terminal frequency $\omega_{\text{s}}$ is small, the damper windings are commonly modeled as inducing a damping power 
\begin{align}
    P_{\text{D,}k} = D_k(\omega_{\text{r,}k} - \omega_{\text{s,}k})\label{eq:dw}
\end{align} 
proportional to the frequency difference across the stator (see \cite[Ch. 5.2]{bialekbook} and \cite[Ch. 6.7]{sauerpai}). Here, $D_k \approx \frac{1}{2}\left(D_{d,k} + D_{q,k}\right)$ is a function of the direct and quadrature axis damping coefficients $D_{d,k}$ and $D_{q,k}$ given by~\cite[p. 174]{bialekbook}
\begin{subequations}\label{eq:damper:fp}
    \begin{align}
        D_{d,k} &= \frac{X_{d,k}^{\prime} - X_{d,k}^{\prime\prime}}{{(X_k+X_{d,k}^{\prime})}^2} \frac{X_{d,k}^{\prime}}{X_{d,k}^{\prime\prime}} T_{d,k}^{\prime\prime}, \\
        D_{q,k} &= \frac{X_{q,k}^{\prime} - X_{q,k}^{\prime\prime}}{{(X_k+X_{q,k}^{\prime})}^2} \frac{X_{q,k}^{\prime}}{X_{q,k}^{\prime\prime}} T_{q,k}^{\prime\prime},
    \end{align} 
\end{subequations}
where $X_k \in \R_{>0}$ represents the armature leakage reactance, $X_{d,k}^{\prime} \in \R_{>0}$ and $X_{q,k}^{\prime} \in \R_{>0}$ represent the $d$- and $q$-axis transient reactances, $X_{d,k}^{\prime\prime} \in \R_{>0}$ and $X_{q,k}^{\prime\prime} \in \R_{>0}$ represent the $d$- and $q$-axis subtransient reactances, and $T_{d,k}^{\prime\prime} \in \R_{>0}$ and $T_{q,k}^{\prime\prime} \in \R_{>0}$ represent the $d$- and $q$-axis subtransient short-circuit time constants. While this model is intuitive, the common derivation either only considers a single-machine infinite bus system \cite[Ch. 5.2]{bialekbook} or replaces $\omega_{\text{s,}k}$ with a constant "rated" system frequency \cite[Ch. 6.7]{sauerpai}. Neither of these approaches can accurately reflect the nature of damper windings in a multi-machine/multi-converter system.

\subsection{Multi-machine damper winding model}
To address this gap, we show that the damping power $P_{\text{D,}k} \in \R$ is proportional to the rate of change of the power $P_{\text{r,net}} \in \R$ flowing from the back EMF bus to the grid. For the sake of brevity, this is shown for a network with only synchronous machines. However, the same result holds without loss of generality for a network which also has converters. First, the network is Kron reduced~\cite{kron} to eliminate all buses apart from these rotor and stator buses. Because the rotors have no direct connection (i.e., edges in $\mathcal{E}_{\text{ac}}$) to any buses other than their corresponding stator buses, all loads in the original network will be mapped to stator buses. Linearizing the AC power flow (see Section~\ref{subsec:acanddctransmission}), the network power flow for a system with $n \in \N$ machines can be written as
\[
    \begin{bmatrix}
        P_{\text{r,net}} \\ -P_{\text{s}}
    \end{bmatrix} = 
    \begin{bmatrix}
        D_B & -D_B \\ -D_B & L_{\text{s}}
    \end{bmatrix}
    \begin{bmatrix}
        \theta_{\text{r}} \\ \theta_{\text{s}}
    \end{bmatrix},
\]
where $P_{\text{r,net}}\in \R^n$ represents the power flowing out of the rotor buses, $P_{\text{s}} \in \R^n$ represents the load mapped from load buses to stator buses, $D_B = \diag\{b_k\}\in \R^{n\times n}$ is a diagonal matrix of the stator impedances, and $L_s \in \R^{n \times n}$ represents the loopy Laplacian matrix corresponding to the network consisting of stator buses (see Fig.~\ref{fig:rs_buses}). Therefore,
\begin{subequations}
    \begin{align}
        \theta_{\text{s}} &= L_{\text{s}}^{-1}D_B\theta_{\text{r}} - L_{\text{s}}^{-1}P_{\text{s}}\label{eq:omega_s}, \\ 
        P_{\text{r,net}} &= \left(D_B -D_B L_{\text{s}}^{-1}D_B\right)\theta_{\text{r}} + D_B L_{\text{s}}^{-1}P_{\text{s}} \label{eq:p_rnet}.
    \end{align}
\end{subequations} 
Note that when there is no load in the AC network, \eqref{eq:omega_s} represents a voltage divider of the AC voltage phase angles at the back EMF terminal\footnote{By the properties of the Kron reduction, $L_{\text{s}}^{-1}D_B$ will be non-negative and row-stochastic~\cite[Lemma~II.2]{kron}, i.e., all entries are non-negative and all rows sum to one. Therefore, $\theta_{\text{s}} = \sum_{i=1}^n \alpha_i \theta_{\text{r,}i}$, where $\alpha_i \in \R_{\geq 0}$ and $\sum_{i=1}^n \alpha_i = 1$.}. Substituting~\eqref{eq:omega_s} into~\eqref{eq:dw} results in
\begin{align*}
    P_{\text{D}} &= D_D\left(\omega_{\text{r}} - \ddt \theta_{\text{s}}\right) \\
    &= D_D\left(\omega_{\text{r}} - L_{\text{s}}^{-1}D_B\omega_{\text{r}} + L_{\text{s}}^{-1}\ddt P_{\text{s}}\right) \\
    &= D_D\left(I - L_{\text{s}}^{-1}D_B\right)\omega_{\text{r}} + D_D L_{\text{s}}^{-1}\ddt P_{\text{s}}, 
\end{align*}
where $P_{\text{D}}\in \R^n$ is the vector of damping powers, and $D_D = \diag\{D_k\}_{k=1}^n \in \R^{n \times n}$. Moreover, defining $\Gamma = D_D{D_B}^{-1} = \diag\{\gamma_k\}_{k=1}^n \in \R^{n \times n}$, we can rewrite $P_{\text{D}}$ as
\begin{figure}[b!]
    \centering
    \tikzset{
block/.style = {draw, fill=white, rectangle, minimum height=3em, minimum width=3em},
tmp/.style  = {coordinate}, 
sum/.style= {draw, fill=white, circle, node distance=1cm},
input/.style = {coordinate},
output/.style= {coordinate,node distance=1cm},
pinstyle/.style = {pin edge={to-,thin,black}
}
}

\begin{tikzpicture}[auto, node distance=2cm,>=latex']
    \node [input, name=rinput] (rinput) {};
    \node [sum, right of=rinput, node distance=1.3cm] (sum1) {};
    \node [sum, right of =sum1] (sum2) {};
    
    \node [block, right of=sum2, node distance=1.5cm] (controller) {$g_{\text{SM,}k}(s)$};
          
   \node [output,right of=controller, node distance=1.2cm] (omega) {};  
   \node [output, right of=omega, node distance = 0.8cm] (yout) {};  

    \node [block, below of=controller,node distance=1.2cm] (ACgen) {$g_{\text{gen,}k}(s)$};    
        \node [output, left of=ACgen, node distance=1.5cm] (pgen) {};
    
    \node [block, below of=sum1, node distance=1.6cm] (dw) {$1 + \gamma_k s$};
    \node [input, below of=dw, node distance=1cm] (pnet) {};

    \draw [->] (rinput) -- node[pos=0.4]{$-P_{\text{load,}k}$} (sum1);
    \draw [->] (pnet) -- node[below, pos=0]{$-p_{\text{net},k}^{\text{ac}}$} (dw);
    \draw [->] (dw) -- node[pos=0.4]{$-p_{\text{net},k}^{\text{ac}} - P_{\text{D},k}$} (sum1);
    \draw [->] (sum1) -- (sum2);
    \draw [->] (sum2) --  (controller);
    \draw [-] (controller) --   (omega);
    
    \draw [-] (ACgen) -- node[pos=0.6]{$p_{\text{gen,}k}$} (pgen);
    \draw [->] (pgen) -- node[pos=0.9]{} (sum2);
    
         \draw [->] (omega) |-  (ACgen);  
         \draw [->] (omega) -- node{$\omega_{\text{r,}k}$} (yout);  

    \end{tikzpicture}
    \caption{Block diagram of the small-signal model of AC generation with a synchronous generator with damper windings at a bus $k$.}\label{fig:ac_dw_block_diag}
\end{figure}
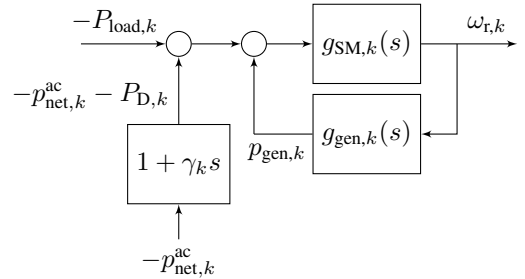
\begin{figure}[b!]
    \centering
    \includegraphics[width=0.9\linewidth]{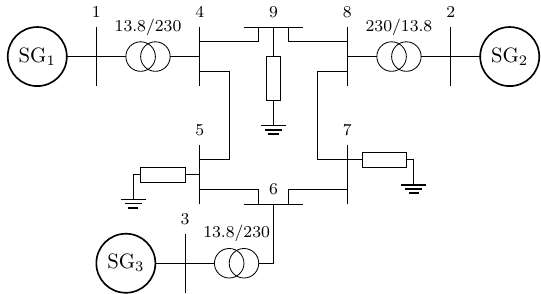}
    \caption{Diagram of the IEEE 9-bus system with three synchronous generators.}\label{fig:9bus}
\end{figure}
\begin{figure*}
    \centering
    \begin{subfigure}{0.3\textwidth}
        \centering
        \includegraphics[width=1\linewidth]{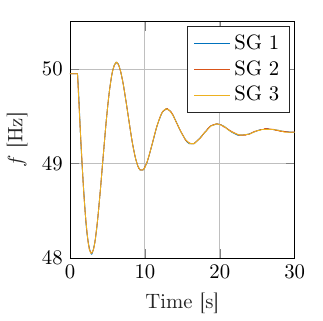}
        \subcaption{EMT simulation}\label{fig:freq_emt_dw}
    \end{subfigure}
    \hfill
    \begin{subfigure}{0.3\textwidth}
        \centering
        \includegraphics[width=1\linewidth]{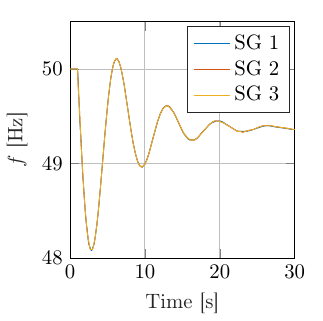}
        \subcaption{Model with damper windings}\label{fig:freq_model_dw}
    \end{subfigure}
    \hfill
    \begin{subfigure}{0.3\textwidth}
        \centering 
        \includegraphics[width=1\linewidth]{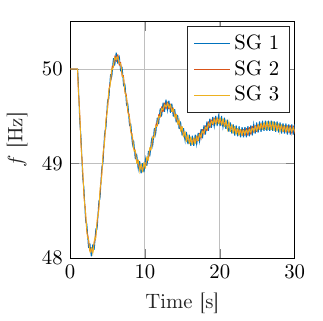}
        \subcaption{Model without damper windings}\label{fig:freq_model_no_dw}
    \end{subfigure}
    \medskip

    \begin{subfigure}{0.3\textwidth}
        \centering
        \includegraphics[width=1\linewidth]{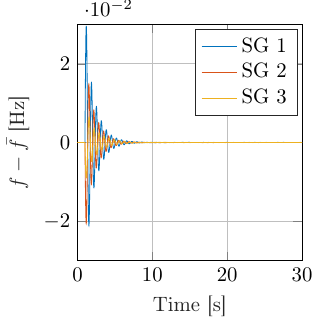}
        \subcaption{EMT simulation}\label{fig:freq_diff_emt_dw}
    \end{subfigure}
    \hfill
    \begin{subfigure}{0.3\textwidth}
        \centering
        \includegraphics[width=1\linewidth]{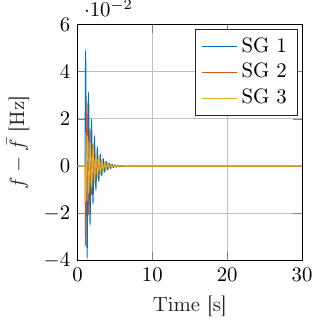}
        \subcaption{Model with damper windings}\label{fig:freq_diff_model_dw}
    \end{subfigure}
    \hfill
    \begin{subfigure}{0.3\textwidth}
        \centering 
        \includegraphics[width=1\linewidth]{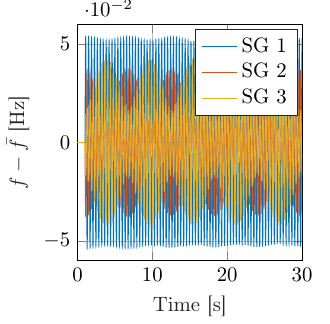}
        \subcaption{Model without damper windings}\label{fig:freq_diff_model_no_dw}
    \end{subfigure}
    \caption{Comparison of the small-signal model with and without the damper winding model and an EMT simulation of the IEEE 9-bus system after a $0.75$~p.u. load step. The plots compare the generator frequencies $f$ and the difference of each generator's frequency to the average frequency $\bar{f}$.}\label{fig:dw_validation}
\end{figure*}
\begin{subequations}
    \begin{align}
        P_{\text{D}} &= \Gamma \left(\left(D_B - D_B L_{\text{s}}^{-1}D_B\right)\omega_{\text{r}} + D_B L_{\text{s}}^{-1}\ddt P_{\text{s}}\right),\\
        P_{\text{D}}&= \Gamma \ddt P_{\text{r,net}}. \label{eq:pdpow}
    \end{align}
\end{subequations}
Note that the $P_{\text{r,net}}$ in \eqref{eq:pdpow} is equivalent to $P^\text{ac}_{\text{net}}$ defined in Section~\ref{sec:model_sm}. Using \eqref{eq:pdpow}, the machine dynamics \eqref{eq:swing_eq} become
\begin{align}
    J_k\omega_0 \ddt \omega_{\text{r,}k} 
    &= P_{\text{m,}k} - P_{\text{load,}k} - (P_{\text{net,}k} + \gamma_k \ddt P_{\text{net,}k}).\label{eq:swing_dw}
\end{align}
While this result was derived for a network with only synchronous machines,~\eqref{eq:swing_dw} holds without loss of generality in a network which also contains converters. The complete model of the dynamics of the synchronous machine, including power generation and dampers windings, is shown in Fig.~\ref{fig:ac_dw_block_diag}.

\subsection{Validation of the Damper Winding Model}
To validate the damper winding model \eqref{eq:pdpow}, we compare the small-signal model of the IEEE 9-bus system in Fig.~\ref{fig:9bus} with an electromagnetic transient (EMT) simulation. The small-signal model is parametrized using the machine parameters given in Table~\ref{Table} in Appendix~\ref{app:ieee9} and standard IEEE 9-bus line parameters. We analyze the system response to a large load step of $0.75$~p.u at Bus 7, comparing the EMT simulation to the small-signal model with the damper winding model \eqref{eq:pdpow} parametrized according to \eqref{eq:damper:fp} and without the damper winding model (i.e., $D=0$). In particular, we examine the response of the frequency of each machine and the deviation of each machine's frequency to the average frequency after the load step (Fig.~\ref{fig:dw_validation}).

While the average frequency response of the small-signal model without damper windings (Fig.~\ref{fig:freq_model_no_dw}) matches the average frequency response of the EMT simulation (Fig.~\ref{fig:freq_emt_dw}), the frequencies of the different machines oscillate against each other without converging (Fig.~\ref{fig:freq_diff_model_no_dw}). Notably, incorporating the damper winding model damps these frequency oscillations (Fig.~\ref{fig:freq_diff_model_dw}), which matches the convergence exhibited by the EMT simulation (Fig.~\ref{fig:freq_diff_emt_dw}). These results indicate that modeling the effects of damper windings as described in Sec.~\ref{sec:dw} allows the small-signal machine model to accurately capture the frequency dynamics of a multi-machine system. We emphasize that modeling the damper windings as absolute damping (i.e., $P_{\text{D},k}=-D \omega_{\text{r},k}$) to the small-signal model \eqref{eq:swing_eq} cannot be used improve the accuracy of the small-signal model.

\section{Interconnected System Model}\label{sec:inter}
In this section, we introduce the frequency-domain model of the overall hybrid AC/DC power system. Here, we again consider a Kron-reduced network with $n \in \N$ buses, where $\mathcal{N} = \{1,\dots,n\}$ is the set of all buses. It is assumed that (i) the network is connected and (ii) that all DC generation is located at buses with converters, so all AC buses in the Kron-reduced network either correspond to generators or converters and all DC buses correspond to converters, which may be attached to a DC resource or network. Consequently, $\mathcal{N}_{\text{ac}} = \mathcal{N}$ is the set of buses connected to AC lines and $\mathcal{N}_{\text{dc}} \subseteq \mathcal{N}$ is the set of buses connected to converters. Moreover, we define $\mathcal{N}_{\text{dc}}^{\text{net}} \subseteq \mathcal{N}_{\text{dc}}$ as the set of DC buses directly connected to a DC network (i.e., connected to at least one other DC bus via a DC line).
We again define $\mathcal{E}_{\text{ac}} \subseteq \mathcal{N}_{\text{ac}} \times \mathcal{N}_{\text{ac}}$ and $\mathcal{E}_{\text{dc}} \subseteq \mathcal{N}_{\text{dc}} \times \mathcal{N}_{\text{dc}}$ to denote the sets of AC and DC lines. 

\subsection{Bus Dynamics}\label{sec:bus_dynamic_model}
The dynamics of each bus are represented by two transfer functions, $g(s)$ and $k(s)$, which together convert from total load to AC frequency.
\subsubsection{Converter Bus Dynamics}
For a DC bus $i \in \mathcal{N}_{\text{dc}}$, we define $g_i(s)$ to be the overall transfer function 
\begin{align*}
    g_i(s) = \dfrac{g_{\text{c,}i}(s)}{1 - g_{\text{c,}i}(s)g_{\text{gen,}i}(s)}
\end{align*}
from the net load mapped to that bus to the DC bus voltage $v_i$ (see Fig.~\ref{fig:dc_blockdiag}). As in Section~\ref{sec:conv_model}, $k_i(s)$ is defined as the transfer function from DC voltage to AC frequency for the converter and is used to model converter controls. 

\subsubsection{AC Bus Dynamics}
For an AC bus $k \in \mathcal{N}_{\text{ac}}$, define $g_k(s)$ as the total transfer function 
\begin{align*}
    g_k(s) = \dfrac{g_{\text{SM,}k}(s)}{1 - g_{\text{SM,}k}(s)g_{\text{gen,}k}(s)}
\end{align*}
from the net load mapped to that bus to the AC frequency. Notably, the effect of the damper windings is not included in $g_k(s)$ but modeled by $k_k(s) = 1 + \gamma_k s$. 

To align with the converter bus dynamics, the order of the transfer functions $g(s)$ and $k(s)$ in Fig.~\ref{fig:ac_dw_block_diag} can be swapped, effectively rescaling the input load disturbance by $k^{-1}(s)$. Note that by swapping the two blocks, the intermediate node between $g(s)$ and $k(s)$, which we call $\tilde{\omega}_{\text{r}}$, no longer directly corresponds to any of the physical quantities in the network but can be thought of as the machine's frequency without the effects of the damper windings while $\omega_{\text{r}}$ is the machine's frequency.

\subsection{Closed-loop Block Diagram}
The overall closed-loop dynamics of the hybrid AC/DC system are shown in Fig.~\ref{fig:acdc_blockdiag}.
\begin{figure}[b!]
    \centering
    \tikzset{
block/.style = {draw, fill=white, rectangle, minimum height=3em, minimum width=3em},
tmp/.style  = {coordinate}, 
sum/.style= {draw, fill=white, circle, node distance=1cm},
input/.style = {coordinate},
output/.style= {coordinate,node distance=1cm},
pinstyle/.style = {pin edge={to-,thin,black}
}
}

\begin{tikzpicture}[auto, node distance=2cm,>=latex']
    \node [input, name=rinput] (rinput) {};
    \node [sum, right of=rinput] (sum1) {};
    \node [sum, below of=sum1, node distance=1.2cm] (sum2) {};    
    
    \node [block, right of=sum1, node distance=1.5cm] (controller) {\scalebox{0.5}{$\!\!\!\begin{array}{ccc} g_1(s)\!\!\!\!\!\!\! & & \\ & \ddots & \\ &  & \!\!\!\!\!\!\!g_n(s) \end{array}\!\!\!$}};
          
    \node [output, right of=controller, node distance=1.5cm] (busmid){};
    
    \node [block, right of=busmid,node distance=1.5cm] (Km) {\scalebox{0.5}{$\!\!\!\begin{array}{ccc} \!k_1(s)\!\!\!\!\!\!\! & & \\ & \ddots & \\ &  & \!\!\!\!\!\!\!k_n(s)\! \end{array}\!\!\!$}};   
    
    \node [output,right of=Km, node distance=1cm] (ybranch) {};      
        \node [output,right of=ybranch, node distance=0.5cm] (yout) {};      

    \node [below of=busmid,node distance=2cm] (system)     
{};
    \node [block, below of=controller,node distance=1.2cm] (DCnet) {$L_{\text{dc}}$};    
        \node [block, below of=busmid,node distance=2.4cm] (ACnet){$\frac{1}{s} L_{\text{ac}}$};

    \draw [->] (rinput) -- node{$-P_{\text{dist}}$} (sum1);
    \draw [->] (sum1) --  (controller);
    \draw [->] (controller) --  node[]{$z = (\tilde{\omega}_{\text{r}}, v)$}  (Km);
    \draw [->] (Km) --  node[]{$\omega_r$}  (yout);
    
    \draw [->] (sum2) --   node[pos=0.9]{$-$}  (sum1);    
    \draw [->] (ACnet) -| node[pos=0.7]{$P^{\text{ac}}_{\text{net}}$} (sum2);
    \draw [->] (DCnet) -- node[]{$P^{\text{dc}}_{\text{net}}$} (sum2);
    
        \draw [->] (busmid) |-  (DCnet);    
\draw [->] (ybranch) |-  (ACnet);    

    \end{tikzpicture}
    \caption{The overall block diagram for the hybrid AC/DC system model.}\label{fig:acdc_blockdiag}
\end{figure}
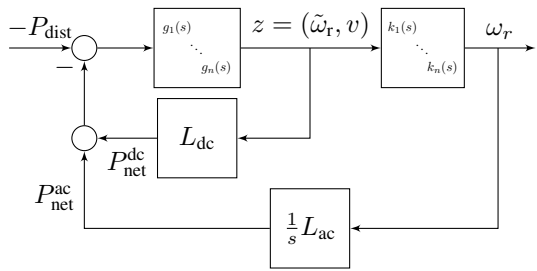
%
Notably, the system input 
\begin{align*}
    P_{\text{dist,}i} = \left\{\begin{array}{lr}
        P_{\text{load,}i}, & i \in \mathcal{N}_{\text{dc}} \\
        k_i^{-1}(s)P_{\text{load,}i}, & i \notin \mathcal{N}_{\text{dc}}
    \end{array} \right.
\end{align*}
now represents the load mapped (by Kron reduction) to each bus  and filtered by the inverse damper winding model for synchronous machine buses. Next, we define $G(s) \in \C^{n \times n}$ and $K(s) \in \C^{n \times n}$ as the diagonal matrices of the individual $g(s)$ and $k(s)$ transfer functions at each bus. Moreover, we define $P_{\text{net}}^{\text{ac}} \in \R^{n}$ as the vector containing the power injection at each AC bus and $P_{\text{net}}^{\text{dc}} \in \R^{n}$ be the power injection at each DC bus. Note that $G(s)$ maps the difference between $P_{\text{dist}}$ and the total network power injection to the vector $z \in \R^n$, where $z_i = v_i$ if $i \in \mathcal{N}$ is a converter bus and $z_i = \tilde{\omega}_{\text{r,}i}$ if $i \in \mathcal{N}$ is a machine bus. Then, $K(s)$ maps $z$ into the vector of AC frequencies $\omega_{\text{r}} \in \R^n$.

The power flow through both the AC and DC networks is again modeled by Laplacian matrices $L_{\text{dc}} \in \R^{n \times n}$ and $L_{\text{ac}} \in \R^{n \times n}$. Notably, $L_{\text{dc}} \in \R^{n \times n}$ is defined as the Laplacian matrix over the graph $(\mathcal{N}, \mathcal{E}_{\text{dc}})$, i.e., $P_{\text{net}}^{\text{dc}} = L_{\text{dc}}z$ and $P_{\text{net},i}^{\text{dc}} = 0$ if $i \notin \mathcal{N}_{\text{dc}}^{\text{net}}$. Similarly, $L_{\text{ac}} \in \R^{n \times n}$ is defined as the Laplacian matrix over the graph $(\mathcal{N}, \mathcal{E}_{\text{ac}})$ and the AC network power flow is modeled as $P_{\text{net}}^{\text{ac}} = L_{\text{dc}} \frac{1}{s} \omega_{\text{r}}$.

\section{Stability Analysis}\label{sec:stability}
This section presents a set of analytical conditions on the bus dynamics of a hybrid AC/DC transmission system to ensure AC frequency stability and DC voltage stability. First, the general hybrid network in Fig.~\ref{fig:acdc_blockdiag} is analyzed for an arbitrary set of bus dynamics $g(s)$ and $k(s)$, and a set of general stability conditions for the node dynamics is put forward with a discussion of their associated interpretations. A simplification of these conditions is then introduced to clarify how these conditions apply to converters on DC networks. 

\subsection{Analytical Stability Conditions}
Let $n_{\mathrm{net}}^{\text{ac}} \in \mathbb{N}$ denote the number of disconnected subnetworks in the AC graph, and $n_{\mathrm{net}}^{\text{dc}} \in \mathbb{N}$ denote the number of disconnected subnetworks in the DC graph. We use 
$\mathcal{N}_{\text{dc}}^i \subseteq$  $\mathcal{N}_{\text{dc}}^{\text{net}}$ to denote the sets of nodes in the 
$i$-th DC subnetwork, where we consider a subnetwork to be a connected component of the DC graph of at least two buses. First, we define a notion of input-output $\mathscr{H}_{\infty}$ stability. A transfer function $F(s) \in \C^{m \times n}$ is $\mathscr{H}_{\infty}$ stable if $\|F(s)\|_{\infty}$ is bounded, where 
\begin{align*}
    \|F(s)\|_{\infty} \coloneqq \sup_{s \in \C_+} \sigma_{\text{max}}\left(F(s)\right)
\end{align*}
and $\sigma_{\text{max}}\left(F(s)\right)$ denotes the largest singular value of $F(s)$. The following lemma is useful for proving $\mathscr{H}_{\infty}$ stability.

\begin{lemma}\label{lemma:real_part} For a given $s \in \C$ and $\kappa \in \R_{>0}$, an invertible matrix $H(s) \in \C^{n\times n}$ will satisfy $\sigma_{\text{max}}\left(H(s)\right) \leq \kappa$
    if \begin{align}\label{eq:re_stability_test}
        \forall x\in \C^n \::\: \|x\| = 1, \quad & \Re\left(x^*H^{-1}(s) x\right) \geq \frac{1}{\kappa}.
    \end{align}
\end{lemma}
A proof is given in Appendix~\ref{app:proofs}. The overall transfer function from $P_{\text{dist}}$ to $\omega_{\text{r}}$ in Fig.~\ref{fig:acdc_blockdiag} can be written as
\begin{align}
    H(s) \coloneqq {\left(G^{-1}(s)K^{-1}(s) + L_{\text{dc}}K^{-1}(s) + \frac{1}{s}L_{\text{ac}}\right)}^{-1},\label{eq:y_tf}
\end{align}
and the transfer function from $P_{\text{dist}}$ to $(\tilde{\omega}_{\text{r}},v)$ 
can be written as $K^{-1}(s)H(s)$. Next, for each $\mathcal{N}_{\text{dc}}^j$, let $\overline{k^{-1}_j}(s) = \frac{1}{|\mathcal{N}_{\text{dc}}^j|}\sum_{i \in \mathcal{N}_{\text{dc}}^j}k_i^{-1}(s)$. 
We define 
\begin{align*}
    S_{\delta} &= \left\{s \in \overline{\C}_+ \: : \: |s| < \delta \right\}, \\
    \Delta(s) &= \max_{j \in \left\{1,\,\dots,\,n_{\text{net}}^{\text{dc}}\right\}}\max_{i \in \mathcal{N}_{\text{dc}}^j} \left|k_i^{-1}(s) - \overline{k^{-1}_j}(s)\right|, \\
    \xi_i(s) &=  \Re\left(g_i^{-1}(s)k_i^{-1}(s)\right).
\end{align*}
We also define $\lambda_{\text{ac}}^{\text{min}}$ and $\lambda_{\text{ac}}^{\text{max}}$ to be the minimum non-zero eigenvalue and maximum eigenvalue of $L_{\text{ac}}$, and equivalently define $\lambda_{\text{dc}}^{\text{min}}$ and $\lambda_{\text{dc}}^{\text{max}}$ for $L_{\text{dc}}$.
Let $L_{\text{dc}}^j$ be the Laplacian matrix of the graph over $(\mathcal{N}, \mathcal{E}_{\text{dc}}^j)$ where $\mathcal{E}_{\text{dc}}^j$ is the set of edges in the $j$-th DC subnetwork. Similarly, we define $\lambda_{\text{dc,}j}^{\text{min}}$ as the smallest non-zero eigenvalue of $L_{\text{dc}}^j$.

To show that the $\mathscr{H}_{\infty}$ norm of $H(s)$ and $K^{-1}(s)H(s)$ are bounded, we apply Lemma~\ref{lemma:real_part} to $H(s)$ and show that $\sigma_{\text{max}}(H(s))$ will be bounded by some $\kappa$ for all $s \in \C_+ \setminus 0$. We decompose $H^{-1}(s)$ into the sum of $G^{-1}(s)K^{-1}(s)$, $L_{\text{dc}}K^{-1}(s)$, and $\frac{1}{s}L_{\text{ac}}$ and separately analyze each of these terms to ensure the total value of $\Re(x^*H^{-1}(s)x)$ will be positive. 

\begin{condition}[\textbf{Bus dynamic conditions}]\label{cond:acdc_universal}
    For some $\delta \in \R_{>0}$, the following conditions hold for every bus $i \in \mathcal{N}$: 
    \begin{enumerate}[leftmargin=*]
        \item $k_i^{-1}(s)$ is $\mathscr{H}_{\infty}$ stable\label{acdc_cond:stable_k}, 
        \item there exists $c \in \R_{>0}$ such that $|g_i^{-1}(s)| \leq c$ holds for all $s \in S_{\delta}$,\label{acdc_cond:stable_ss_g}
        \item for all $s \in \overline{\C}_+\setminus 0$ it holds that $\xi_i(s) > 0$,\label{acdc_cond:pr_freq}
        \item for all $s \in \overline{\C}_+$ it holds that $\frac{1}{n}\sum_{i \in \mathcal{N}} \xi_i(s) > 0$. \label{acdc_cond:avg_stab}
    \end{enumerate}
\end{condition}

\begin{condition}[\textbf{Coherent dynamics of networked DC buses}]\label{cond:acdc_dc_network} \phantom{.} \vspace{-1em}
    \begin{enumerate}[leftmargin=*]
        \item \label{acdc_cond:avg_k_spr}$\Re \left(\overline{k^{-1}_j}(s)\right) > 0$ for all $s \in \overline{\C}_+$ and all $j \in \left\{1,\,\dots,\,n_{\text{net}}^{\text{dc}}\right\}$, and there exists $c_1 \in \R_{>0}$ such that $\Re \left(\overline{k^{-1}_j}(s)\right) \geq c_1$ for all $s \in S_{\delta}$, 
        \item \label{acdc_cond:k_ss_conv} there exists $c_2 \in \R_{>0}$ such that $\Delta(s) \leq c_2 |s|$ holds for all $s \in S_{\delta}$,
        \item for all $s \in \overline{\C}_+ \setminus 0$ it holds that
        \begin{align*}
            4\left(\min_{i\in \mathcal{N}}\xi_i(s)\right)\left(\min_{j \in \left\{1,\,\dots,\,n_{\text{net}}^{\text{dc}}\right\}} \Re \left(\overline{k^{-1}_j}(s)\right) \frac{\lambda_{\text{dc,}j}^{\text{min}}}{\lambda_{\text{dc}}^{\text{max}}}\right) > \lambda_{\text{dc}}^{\text{max}} \Delta^2(s).
        \end{align*}\label{acdc_cond:bounding_k_diff}
    \end{enumerate}
\end{condition}

Before stating the main result, we provide interpretations of the stability conditions. Notably, Condition~\ref{cond:acdc_universal} applies to every bus on the network. Condition~\ref{cond:acdc_universal}.\ref{acdc_cond:stable_k} requires that all controls and damper windings satisfy $k_i^{-1}(s) \in \mathscr{H}_\infty$, ensuring that the DC voltages will always be stable if the AC bus frequencies are stable. Condition~\ref{cond:acdc_universal}.\ref{acdc_cond:stable_ss_g} ensures that no $g_i(s)$ has a zero at $s=0$. If this condition is not met then the coherent dynamics $g(s)$, which are given by the harmonic mean of the functions $g_i(s)$~\cite{avg_mode}, do not admit a stable steady state. Condition~\ref{cond:acdc_universal}.\ref{acdc_cond:pr_freq} requires that each bus has strictly positive real part for $s \in \overline{\C}_+ \setminus 0$, which ensures that each bus contributes to damping transient dynamics. Similarly, Condition~\ref{cond:acdc_universal}.\ref{acdc_cond:avg_stab} requires that the coherent dynamics of $g(s)k(s)$ have strictly positive real part for all $s\in \overline{\C}_+$ (i.e., contribute damping) which is one of the properties of a strictly passive system~\cite[Ch. 6]{khalil2002nonlinear}.

Additionally, Condition~\ref{cond:acdc_dc_network} applies to all DC subnetworks individually. Condition~\ref{cond:acdc_dc_network}.\ref{acdc_cond:avg_k_spr} requires that the coherent dynamics of the converter controls (represented by $k_i(s)$) on each DC subnetwork will have strictly positive real part for all $s\in \overline{\C}_+$ (i.e., contribute damping). Moreover, by ensuring $\Re \left(\overline{k^{-1}_j}(s)\right)$ is lower bounded by a strictly positive number for $s \in S_{\delta}$, Condition~\ref{cond:acdc_dc_network}.\ref{acdc_cond:avg_k_spr} ensures the coherent converter control dynamics of each DC network will have finite steady-state gain. Condition~\ref{cond:acdc_dc_network}.\ref{acdc_cond:k_ss_conv} requires that the controls of each converter on a given DC subnetwork converge to the coherent dynamics in steady state. 

Condition~\ref{cond:acdc_dc_network}.\ref{acdc_cond:bounding_k_diff} requires that the converter controls on each DC subnetwork are sufficiently coherent by providing an upper bound on $\Delta(s)$. Note that $\xi_i(s)$ corresponds to the amount of positive damping provided by the bus dynamics of bus $i$, and $\Re \left(\overline{k^{-1}_j}(s)\right)$ corresponds to the positive damping provided by the coherent control dynamics of the $j$-th DC subnetwork. If the amount of damping provided by the bus dynamics or by the coherent control dynamics of the DC subnetworks are reduced, the control coherency requirement in Condition~\ref{cond:acdc_dc_network}.\ref{acdc_cond:bounding_k_diff} becomes stricter. The coherency requirement is also a function of the DC network parameters. Specifically, $\lambda_{\text{dc,}j}^\text{min}$ corresponds to the strength of the weakest connection in the $j$-th DC subnetwork, and $\lambda_{\text{dc}}^\text{max}$ can be understood as a metric of the connectivity of the most tightly connected node in the DC network. If the ratio $\lambda_{\text{dc,}j}^\text{min}/\lambda_{\text{dc}}^\text{max}$ becomes smaller, or if the DC network becomes more tightly connected (i.e. $\lambda_{\text{dc}}^\text{max} \to \infty$), the coherency requirement for the converter controls on each DC subnetwork will become stricter.

\begin{theorem}[\textbf{Frequency and DC voltage stability}]\label{thm:acdc_stability}
    Consider $K(s)$ and $G(s)$ such that Conditions~\ref{cond:acdc_universal} and~\ref{cond:acdc_dc_network} hold. Then, if $L_{\text{ac}}$ and $L_{\text{dc}}$ together represent the graph of a connected AC/DC network, the transfer functions $H(s)$ and $K^{-1}(s)H(s)$ are $\mathscr{H}_{\infty}$ stable.
\end{theorem}
A proof is given in Appendix~\ref{app:proofs}.

\subsection{Simplified Stability Conditions}
Condition~\ref{cond:acdc_dc_network} can be simplified if we restrict the controls for converters connected to DC networks such that for all $j \in \{1,\,\dots,\,n_{\text{dc}}\}$ and all $i \in \mathcal{N}_{\text{dc}}^j$, $k_i(s)$ can be expressed as 
\begin{align}
    k_i(s) = \alpha_j + s\beta_i, \label{eq:k_form}
\end{align} 
where  $\alpha_j,\, \beta_i \in \R_{> 0}$. This condition is sufficient to ensure Conditions~\ref{cond:acdc_dc_network}.\ref{acdc_cond:avg_k_spr} and~\ref{cond:acdc_dc_network}.\ref{acdc_cond:k_ss_conv} are satisfied. 

Requiring identical $\alpha_j$ for each bus in a DC subnetwork effectively requires that the steady state gain from DC voltage to AC frequency must be the same for each converter on every DC subnetwork. This ensures that, in steady state, there will not be circulating power flows between AC and DC subnetworks~\cite{energybalance}. 
Moreover, in this representation, the differences between $\beta_i$ for all $i \in \mathcal{N}_{\text{dc,}j}$ can be understood as a metric of the control coherency of the $j$-th DC subnetwork: to reduce $\Delta(s)$, we can tune the converter controls to reduce the difference between $\beta_i$ across each DC subnetwork.
In particular, if each $k_i(s)$ is the same on a given DC subnetwork, Condition~\ref{cond:acdc_dc_network}.\ref{acdc_cond:bounding_k_diff} will automatically be satisfied.

In summary, to improve the stability margin of the system, one can increase the passivity margin $\xi_i(s)$ of each bus in the system, increase the passivity margins of the coherent  dynamics $\bar{k}(s)$ for each DC subnetwork, increase the coherency of the different $k_i(s)$ for each DC subnetwork, or make the entire DC system less strongly connected. This supports recent observations in the literature that grid-forming converters are prone to instability under strong coupling~(see \cite{9929370,8625604} and the references therein).

\section{Application to Common Power Generation and Energy Conversion Technologies}\label{sec:applicationexample}
In this section, we demonstrate the utility of the derived conditions by applying them to a variety of common generation and energy conversion technologies. In particular, we show that the conditions can be used to establish stabililty of network configurations that are known to be stable. 
We note that the stability conditions can be divided into two categories: (i) conditions on each bus (Conditions~\ref{cond:acdc_universal}.\ref{acdc_cond:stable_k} --~\ref{cond:acdc_universal}.\ref{acdc_cond:pr_freq}), and (ii) conditions on the interconnected system (Condition~\ref{cond:acdc_universal}.\ref{acdc_cond:avg_stab} and Condition~\ref{cond:acdc_dc_network}). In this section, we illustrate the bus-level stability conditions for individual bus dynamics and then discuss the implications of the system-level stability conditions.

\subsection{Bus-Level Stability Conditions}
In this section, we examine how several common generation and energy conversion technologies can be represented in the hybrid system model and apply the bus-level conditions. 

\subsubsection{Synchronous Generator with Turbine-Governor}
A synchronous machine with a turbine-governor system and damper windings is described by 
\begin{align*}
    g(s) &= \frac{1 + \tau s}{J\omega_0\tau s^2 + J\omega_0s + k_{\text{g}}}, \\
    k(s) &= 1 + \gamma s.
\end{align*}
Note that $k^{-1}(s)$ is an $\mathscr{H}_\infty$ stable transfer function, $|g^{-1}(s)|$ converges to $k_g$ as $s \to 0$, and $g^{-1}(s)k^{-1}(s)$ has positive real part for all $s \in \overline{\C}_+$, so the bus-level conditions are satisfied.

\subsubsection{Synchronous Condenser}
Like the synchronous generator, the synchronous condenser has damper windings, so $k(s) = 1 + \gamma s$, satisfying Condition~\ref{cond:acdc_universal}.\ref{acdc_cond:stable_k}. However, a synchronous condenser has no turbine-governor system, so 
\begin{align*}
    g(s) = \frac{1}{J\omega_0s}.
\end{align*}
Here, $g^{-1}(s)$ converges to zero for $s \to 0$, satisfying Condition~\ref{cond:acdc_universal}.\ref{acdc_cond:stable_ss_g}. Notably, for the synchronous condenser, $\Re\left(g^{-1}(s)\right)$ will be zero when $\Re(s) = 0$. However, incorporating the damper winding dynamics, we see that $\Re\left(g^{-1}(s)k^{-1}(s)\right) > 0$ for all $s \in \overline{\C}_+\setminus 0$, satisfying Condition~\ref{cond:acdc_universal}.\ref{acdc_cond:pr_freq}. This indicates the benefit of the reduced-order damper winding model presented in Section~\ref{sec:dw} in demonstrating the stability of synchronous condensers for arbitrary AC network topologies. In contrast, for a three-bus swing equation model with one converter and two synchronous condensers that neglects the damper windings, there always exist line parameters so that the multi-machine swing equation model is unstable~\cite[Ex.~1]{powerbalance}.

\subsubsection{PV System with Dual-Port GFM Control}
Consider a single-stage PV system containing a two-level converter with dual-port GFM control directly connected to a PV array. The dual-port GFM control is represented by 
\begin{align*}
    k(s) = k_{\omega} + k_{p}s,
\end{align*} 
where $k_{\omega} \in \R_{>0}$ and $k_{p} \in \R_{>0}$ are proportional and derivative gains of the control~\cite{universaldualport}. Note that $k^{-1}(s)$ is $\mathscr{H}_\infty$ stable.

If the PV system is performing MPP tracking, $g(s) = 1/C_{\text{dc}}s.$ This case is analogous to the synchronous condenser, i.e., $g^{-1}(s)$ converges to zero for $s \to 0$ and $g^{-1}(s)k^{-1}(s)$ has positive real part for all $s \in \overline{\C}_+\setminus 0$. If the PV system operates at a voltage above its MPP voltage, then
\begin{align*}
    g(s) = \frac{1}{C_{\text{dc}}s + k_{\text{pv}}}.
\end{align*}
This is analogous to the synchronous generator with the turbine-governor system, i.e., $g^{-1}(s)$ converges to a finite value for $s \to 0$ and $g^{-1}(s)k^{-1}(s)$ has positive real part for all $s \in \overline{\C}_+$. For both PV operating modes, the bus level stability conditions are satisfied. 

\subsubsection{VSC-HVDC with Dual-Port GFM Control}
For an HVDC system, there is no direct generation response, so $g(s) = 1/{C_{\text{dc}}s}$. The dual-port GFM control is again represented by $k(s) = k_{\omega} + k_{p}s$, i.e., identical to the PV system at the MPP, and the bus-level conditions are satisfied.

\subsubsection{Battery Energy Storage with GFM VSM Control}\label{sec:droop_application}
Consider a battery energy storage system with GFM VSM control. Assume the battery is controlled so its output power responds proportionately to DC voltage deviations, i.e., $g_{\text{gen}}(s) = -k_{\text{batt}}$ for some finite $k_{\text{batt}} \in \R_{\geq 0}$. Then, the DC-link capacitor dynamics are given by
\begin{align*}
    g(s) = \frac{1}{C_{\text{dc}}s + k_{\text{batt}}}
\end{align*}
and $g^{-1}(s)$ will converge to a finite value for $s \in S_{\delta}$ for small enough $\delta$, satisfying Condition~\ref{cond:acdc_universal}.\ref{acdc_cond:stable_ss_g}. Moreover, the dynamics of GFM VSM control can be modeled as
\begin{align*}
    k(s) = \frac{m_p(C_{\text{dc}}s + k_{\text{batt}})}{Ts + 1},
\end{align*}%
to obtain $g(s)k(s) = m_p/(Ts + 1)$, where $T \in \R_{\geq 0}$ is proportional to the virtual inertia constant and $m_p \in \R_{>0}$ is the active power-frequency droop coefficient. GFM droop control can be modeled as a special case of VSM control with $T=0$~\cite{droop_schiffer_2013,equivalence_darco_2014}. Moreover the small-signal model of dVOC control is equivalent to droop control for transmission systems~\cite{seo_DVOC_2019}. The real part of $g^{-1}(s)k^{-1}(s)$ will always be positive for $s \in \overline{\C}_+$, so Condition~\ref{cond:acdc_universal}.\ref{acdc_cond:pr_freq} is satisfied. If $k_{\text{batt}} > 0$, $k^{-1}(s)$ will satisfy Condition~\ref{cond:acdc_universal}.\ref{acdc_cond:stable_k}. However, if $k_{\text{batt}} = 0$, $k^{-1}(s)$ will no longer be $\mathscr{H}_{\infty}$ stable. This reflects that AC GFM controls that neglect the DC voltage are unable to stabilize the DC voltage unless the converter is paired with a power source that stabilizes the DC voltage.

\subsubsection{VSC with GFL control}
Consider a VSC implementing grid-following control (using a synchronous reference frame phase-locked loop) with frequency droop. The dynamics from net power to AC frequency for this control can be modeled as \begin{align*}
    g(s)k(s) = \frac{s^2 + k_p s + k_i}{k_p s + k_i} \frac{\tau_d s + 1}{D},
\end{align*}
where $k_p, k_i \in \R_{>0}$ are the phase-locked loop gains, $D \in \R_{>0}$ is the frequency damping constant, and $\tau_d \in \R_{>0}$ is the time constant of a filter for measuring frequency~\cite[Eq. 7]{bui_input-output_2024}. This function cannot be passive and will have negative real part for some $s \in \overline{\C}_+$. Therefore, Condition~\ref{cond:acdc_universal}.\ref{acdc_cond:pr_freq} is not satisfied, and our stability conditions cannot certify the stability of a network with a GFL VSC.

\subsection{System-wide Stability Conditions}
\subsubsection{Stable Average Dynamics}
Condition~\ref{cond:acdc_universal}.\ref{acdc_cond:avg_stab} uses the harmonic mean model to ensure system-wide stability for the average mode of the transfer functions $g_i(s)k_i(s)$. Because Condition~\ref{cond:acdc_universal}.\ref{acdc_cond:pr_freq} already requires that each $g_i^{-1}(s)k_i^{-1}(s)$ has non-negative real part for all $s \in \overline{\C}_+\setminus 0$, Condition~\ref{cond:acdc_universal}.\ref{acdc_cond:avg_stab} can be understood as the additional requirement that the coherent dynamics are strictly passive at steady state. As seen from the previous examples, this corresponds to requiring that at least one bus on the system contains a resource that adjusts its power generation in response to imbalances in the AC frequency or DC voltage. If the system contains no DC networks, satisfying this condition and the bus-level conditions presented in the previous section suffices for proving stability. In the next section, we address the stability requirements for converters connected to a DC network.

\subsubsection{DC Network Stability}
In this section we apply Condition~\ref{cond:acdc_dc_network} to characterize the stability of different converter controls when connected to a DC network. First, we note that dual-port GFM control is of the form given in~\eqref{eq:k_form}. Therefore, if every converter on a given DC network is using dual-port GFM control with identical $k_i(0)$ for each converter, Conditions~\ref{cond:acdc_dc_network}.\ref{acdc_cond:avg_k_spr} and~\ref{cond:acdc_dc_network}.\ref{acdc_cond:k_ss_conv} will be satisfied. Additionally, a droop controlled VSC paired with a responsive resource satisfies~\eqref{eq:k_form} if $T = 0$. However, in general, VSM control with $T > 0$ does not satisfy Condition~\ref{cond:acdc_dc_network}.\ref{acdc_cond:k_ss_conv}. This reflects previous findings that applying AC GFM control that neglects the DC voltage to converters connecting AC and DC networks results in stability conditions that are highly sensitive to parameter and topology changes~\cite{singleportroleassignment,singleportunstable}.

Condition~\ref{cond:acdc_dc_network}.\ref{acdc_cond:bounding_k_diff} requires that the transfer functions $k_i(s)$ do not significantly deviate from their average within each DC subnetwork, which corresponds to frequency coherency for the converters. For a more tightly connected network (i.e., large $\lambda_{\text{dc}}^{\text{max}}$), increased coherency is required to ensure stability. For a DC network with all converters using dual-port GFM control, this corresponds to decreasing the difference between derivative control gains of the converters.

\section{Conclusion}\label{sec:conclusion}
In this article, we developed small-signal stability conditions for hybrid AC/DC power systems that allow us to consider common power conversion, power generation, and power transmission technologies. First we reviewed common models for AC and DC transmission, machine-interfaced generation (e.g., turbine-governor systems, wind turbines) and converter-interfaced generation (e.g., solar PV) and transmission (i.e., VSC-HVDC). 
Next, we developed a novel reduced-order damper winding model that accurately captures the synchronizing effect of damper windings in multi-machine systems. The damper winding model is validated by comparing the small-signal model to an EMT simulation of the IEEE 9-bus system. 
The main contribution of this work is a compact frequency domain representation of hybrid AC/DC systems and associated stability conditions that can be divided into bus-level and system-level conditions. We illustrated that the bus-level conditions apply to a wide range of technologies (e.g., synchronous generators, synchronous condensers, grid-forming renewables and energy storage). 
Moreover, the system-level conditions establish that hybrid AC/DC systems combining a wide range of devices are stable independently of the network topology provided that the frequency response of generators and converters is sufficiently coherent relative to the network coupling strength. 

The small-signal analysis used in this work does not fully reflect the non-linear power system dynamics. While the properties of transmission systems make linearization justified, our analysis does not account for the nonlinear dynamics of grid-connected power converters near their limits (e.g., current and voltage limits). Methods such as the Popov criterion~\cite{PM2019} or integral quadratic constraints~\cite{Lessard_IQC_2016} can be used to extend our analysis and account for nonlinearities. Our analysis also neglects AC voltage magnitude dynamics. While this is typically justified for transmission systems, incorporating AC voltage dynamics to analyze distribution systems is seen as an interesting topic of future work. Finally, our analysis also neglects the dynamics of the converter inner loops and line dynamics. While singular perturbation arguments are commonly used to justify using reduced-order models that neglect inner loops and line dynamics for AC systems, future work should investigate timescale separation for hybrid AC/DC networks. Other interesting topics for future work include in-depth validation of the proposed damper winding model and explicitly incorporating transmission line dynamics into the stability analysis framework. 

\appendix
\subsection{Laplacian matrices of AC and DC networks}\label{app:lap}
The entries $\tilde{L}_{\text{ac},ij}$ of $\tilde{L}_{\text{ac}}$ are given by
\begin{align*}
    \tilde{L}_{\text{ac},ij} = \left\{\begin{array}{lr}
    -b_{ij}, & (i,j) \in \mathcal{E}_{\text{ac}}, i \neq j \\
    \sum_{(i,k) \in \mathcal{E}_{\text{ac}}} b_{ik}, & i = j \\
    0, & (i,j) \notin \mathcal{E}_{\text{ac}}, i \neq j
    \end{array} \right. .
\end{align*}
The entries $\tilde{L}_{\text{dc},ij}$ of $\tilde{L}_{\text{dc}}$ are given by
\begin{align*}
    \tilde{L}_{\text{dc},ij} = \left\{\begin{array}{lr}
    -g_{ij}, & (i,j) \in \mathcal{E}_{\text{dc}}, i \neq j \\
    \sum_{(i,k) \in \mathcal{E}_{\text{dc}}} g_{ik}, & i = j \\
    0, & (i,j) \notin \mathcal{E}_{\text{dc}}, i \neq j
    \end{array} \right. .
\end{align*}

\subsection{Parameters of the IEEE 9-bus system}\label{app:ieee9}
{\fontfamily{ptm}\selectfont
    \begin{table}[h!!!]
        \caption{Parameters of the IEEE 9-bus system.\label{Table}}
        \begin{minipage}[c]{\linewidth}
        \centering
        \hspace{-5mm}
        \scalebox{0.74}{
            {\renewcommand{\arraystretch}{1.3}
                \begin{tabular}[]{|c|c||c|c||c|c|}
                    \hline
                    \rowcolor{light-gray}
                    \multicolumn{6}{|c|}{Base values}\\
                    \hline
                    $S_\text{b}$ & $100$ MVA & $v_\text{b}$ & $230$ kV& $\omega_\text{b}$ & $2\pi50$ rad/s\\
                    \hline
                    \rowcolor{light-gray}
                    \multicolumn{6}{|c|}{{MV/HV transformer}}
                    \\ \hline
                    $S_\text{r}$ & $210$ MVA & $v_1$ & $13.8$ kV& $v_2$ & $230$ kV\\\hline
                    $R_1=R_2$ & $0.0027$ p.u. & $L_1=L_2$ & $0.08$ p.u.& $R_\text{m}=L_\text{m}$ & $500$ p.u.\\\hline
                    \rowcolor{light-gray}
                    \multicolumn{6}{|c|}{{synchronous machine (SM)}}\\ \hline
                    $S_\text{r}$ & $100$ MVA & $v_\text{r}$ & $13.8$ kV& $H$ & $3.7$ s \\\hline
                    $k_g$ & $20$ p.u. & $\tau$ & $3$ s & $X$ & $0.114$ p.u. \\\hline
                    $X_d^{\prime}$ & $0.104$ p.u. & $X_d^{\prime\prime}$ & $3.43\times 10^{-4}$ p.u. & $T_d^{\prime\prime}$ & $9.36 \times 10^{-5}$ p.u. \\ \hline
                    $X_q^{\prime}$ & $0.360$ p.u. & $X_q^{\prime\prime}$ & $7.24\times 10^{-4}$ p.u. & $T_q^{\prime\prime}$ & $1.17 \times 10^{-4}$ p.u.\\ \hline
        \end{tabular}}}
        \end{minipage}
\end{table}}

\subsection{Preliminary Results}

\begin{lemma}\label{lemma:asym_eig_bound}
    For a positive semidefinite matrix $M \in \R^{n\times n}$ and a diagonal matrix $D = \diag\{d_i\}_{i=1}^n \in \C^{n \times n}$,
    \begin{align*}\Re(y^* MDx) \leq \|x\|\|y\|\lambda^{\textup{max}}_{M}\left(\max_i |d_i|\right) \qquad \forall x,\, y \in \C^n,\end{align*}
    where $\lambda^{\textup{max}}_{M}$ is the largest eigenvalue of $M$.
\end{lemma}
\begin{proof}
    \begin{align*}
        \Re(y^* MDx) &\leq \|y^* MDx\| \\ &\leq \|y\|\|M\| \|D\|\|x\| \\ &= \|x\|\lambda^{\text{max}}_{M}\left(\max_i |d_i|\right)\|y\|.
    \end{align*}
\end{proof}
\subsection{Proof of the Main Results}\label{app:proofs}
\emph{Proof of Lemma~\ref{lemma:real_part}.}
    The condition $\sigma_{\text{max}}\left(H(s)\right) < \kappa$ is equivalent to $\sigma_{\text{min}}\left(H^{-1}(s)\right) > \frac{1}{\kappa}$.
    Note that by~\cite[Prop. III.5.1]{bhatia1996matrix}, \[\sigma_{\text{min}}\left(H^{-1}(s)\right) \geq  \lambda_{\text{min}}\left(\frac{1}{2}\left(H^{-1}(s) + (H^{-1}(s))^*\right)\right).\]
    By definition,
    \begin{align*}
        \lambda_{\text{min}}&\left(\frac{1}{2}\left(H^{-1}(s) + (H^{-1}(s))^*\right)\right)\\ 
        &=\min_{\substack{x\in \C^n \\\|x\| = 1}} x^*\left(\frac{1}{2}(H^{-1}(s) + (H^{-1}(s))^*)\right)x \\
        &= \min_{\substack{x\in \C^n \\\|x\| = 1}} \frac{1}{2}\left(x^*H^{-1}(s)x + x^*(H^{-1}(s))^*x\right) \\ 
        &= \min_{\substack{x\in \C^n \\\|x\| = 1}}\Re\left(x^*H^{-1}(s)x\right).
    \end{align*}
    Therefore, to prove $\sigma_{\text{max}}\left(H(s)\right) < \kappa$, it is sufficient that $\Re\left(x^*H^{-1}(s)x\right) > \frac{1}{\kappa}$ holds for all $x\in \C^n$ such that $ \|x\| = 1$. \qed{}

\emph{Proof of Theorem~\ref{thm:acdc_stability}.}
    By Condition~\ref{cond:acdc_universal}.\ref{acdc_cond:stable_k}, establishing stability of $H(s)$ (i.e., \eqref{eq:y_tf}) is sufficient for proving  stability of $K^{-1}(s)H(s)$. To this end, we define the diagonal matrix $\overline{K^{-1}}(s)$ with $\overline{K^{-1}_{ii}}(s) = \overline{k^{-1}_j}(s)$ if $i \in \mathcal{N}_{\text{dc}}^j$, and $\overline{K^{-1}_{ii}}(s) = k_i^{-1}(s)$ if node $i$ is a purely AC node. Moreover, we define $K^{\Delta}(s) = K^{-1}(s) - \overline{K^{-1}}(s)$.
    Decomposing $L_{\text{dc}}$ into the Laplacian of each DC subnetwork, we obtain \[L_{\text{dc}}\overline{K^{-1}}(s) = \sum\nolimits_{j = 1}^{n_{\text{net}}^{\text{dc}}} \overline{k_j^{-1}}(s) L_{\text{dc}}^j.\]
    Next, consider an arbitrary $x \in \C^n$ where $\|x\| = 1$. Then 
    \begin{align}\label{eq:acdc_re_x}
        \Re\left(x^*\left(G^{-1}(s)K^{-1}(s) + L_{\text{dc}}K^{-1}(s) + \frac{1}{s}L_{\text{ac}}\right)x \right)
    \end{align}
    \begin{figure*}[t!!]
        \begin{align}
            \begin{bmatrix}
                \|x\| \\ \|x\|_{\text{N}\left(L_{\text{dc}}\right)} 
            \end{bmatrix}^*
            \underbrace{\begin{bmatrix}
                \min\limits_{i \in \mathcal{N}} \xi_i(s) & -\frac{1}{2}\lambda_{\text{dc}}^{\text{max}}\Delta(s) \\
                -\frac{1}{2}\lambda_{\text{dc}}^{\text{max}}\Delta(s) & \min\limits_{j \in \left\{1,\,\dots,\,n_{\text{net}}^{\text{dc}}\right\}} \Re \left(\overline{k^{-1}_j}(s)\right) \lambda_{\text{dc,}j}^{\text{min}}
            \end{bmatrix}}_{\coloneqq M(s)}
            \begin{bmatrix}
                \|x\| \\ \|x\|_{\text{N}\left(L_{\text{dc}}\right)} 
            \end{bmatrix}\label{eq:matrixform}
        \end{align}
        \hrule       
    \end{figure*}
    can be rewritten as 
    \begin{multline*}
        \Re\left(x^*G^{-1}(s)K^{-1}(s)x\right) + \Re\left(\frac{1}{s}x^*L_{\text{ac}}x\right) + \\ \sum_{j = 1}^{n_{\text{net}}^{\text{dc}}} \Re\left(\overline{k_j^{-1}}(s) x^*L_{\text{dc}}^j x\right) + \Re\left(x^*L_{\text{ac}}K^{\Delta}(s)x\right).
    \end{multline*}
    This expression can be lower bounded by lower bounding each term. To this end, let $N(M)$ denote the nullspace of a matrix $M$, and let $\|x\|_{V}$ denote the distance from $x$ to a subspace $V$. Moreover, we define $x_{V}^{\perp}$ as the component of $x$ orthogonal to $V$.
%
    First, 
    \begin{align*}
        \Re(x^*G^{-1}K^{-1}(s)x) &\geq \min_{i \in \mathcal{N}}\Re\left(g_i^{-1}(s)k_i^{-1}(s)\right)\|x\|^2 \\
        &\geq \min_{i \in \mathcal{N}} \xi_i(s)\|x\|^2.
    \end{align*}
    Second, because $L_{\text{ac}}$ is a positive semidefinite real matrix, \begin{align*}\Re\left(\frac{1}{s}x^*L_{\text{ac}}x\right) \geq \Re\left(\frac{1}{s}\right)\lambda_{\text{ac}}^{\text{min}}\|x\|_{\text{N}\left(L_{\text{ac}}\right)}^2,
    \end{align*}
    where $\lambda_{\text{ac}}^{\text{min}} > 0$. For $s \in \overline{\C}_+$, $\Re\left(\frac{1}{s}\right) \geq 0$. Thus, we obtain \[\Re\left(\frac{1}{s}x^*L_{\text{ac}}x\right) \geq 0.\]
    Third, because each $L_{\text{dc}}^j$ is real, it follows that
    \begin{align*}\Re\left(\overline{k_j^{-1}}(s) x^*L_{\text{dc}}^j x\right) = \Re\left(\overline{k_j^{-1}}(s)\right) x^*L_{\text{dc}}^j x.\end{align*} Using the same approach used for the AC case results in
    \begin{align*}
        \sum_{j = 1}^{n_{\text{net}}^{\text{dc}}} \Re\left(\overline{k_j^{-1}}(s) x^*L_{\text{dc}}^j x\right)  &\geq \min_{j \in \left\{1,\,\dots,\,n_{\text{net}}^{\text{dc}}\right\}} \Re \left(\overline{k^{-1}_j}(s)\right) \lambda_{\text{dc,}j}^{\text{min}}\|x\|_{\text{N}\left(L_{\text{dc}}\right)}^2.
    \end{align*}
    Finally, note that \[\Re\left(x^*L_{\text{dc}}K^{\Delta}(s)x\right) = \Re\left({\left(x_{\text{N}(L_{\text{dc}})}^{\perp}\right)}^*L_{\text{dc}}K^{\Delta}(s)x\right).\] Therefore, by Lemma~\ref{lemma:asym_eig_bound}, it follows that
    \begin{align*}
        \Re&\left({\left(x_{\text{N}(L_{\text{dc}})}^{\perp}\right)}^*L_{\text{dc}}K^{\Delta}(s)x\right) \geq -\lambda_{\text{dc}}^{\text{max}}\Delta(s)\|x\|_{\text{N}\left(L_{\text{dc}}\right)}\|x\|.
    \end{align*}

    It follows that~\eqref{eq:acdc_re_x} can be lower bounded by~\eqref{eq:matrixform}. Moreover,~\eqref{eq:matrixform} is positive if $M(s)$ is positive definite. 
    Using the Schur complement and substituting the above equation, $M(s)$ is positive definite if and only if $\min\limits_{j \in \left\{1,\,\dots,\,n_{\text{net}}^{\text{dc}}\right\}} \Re \left(\overline{k^{-1}_j}(s)\right) \lambda_{\text{dc,}j}^{\text{min}}> 0$, $\min_{i\in \mathcal{N}}\xi_i(s) > 0$, and \[
        4\left(\min_{i\in \mathcal{N}}\xi_i(s)\right)\left(\min_{j \in \left\{1,\,\dots,\,n_{\text{net}}^{\text{dc}}\right\}} \Re \left(\overline{k^{-1}_j}(s)\right) \lambda_{\text{dc,}j}^{\text{min}}\right) > \left(\lambda_{\text{dc}}^{\text{max}} \Delta(s)\right)^2.
    \]%
    The first inequality is satisfied by Condition~\ref{cond:acdc_dc_network}.\ref{acdc_cond:avg_k_spr} and by the fact that $\lambda_{\text{dc,}j}^{\text{min}}> 0$ by construction. 
    The second inequality is satisfied by Condition~\ref{cond:acdc_universal}.\ref{acdc_cond:pr_freq}. The third inequality is satisfied by
    Condition~\ref{cond:acdc_dc_network}.\ref{acdc_cond:bounding_k_diff} for all $s \in \overline{\C}_+\setminus 0$. Therefore, for all $s \in \overline{\C}_+ \setminus 0$, the expression in~\eqref{eq:acdc_re_x} is positive for any $x \in \C^n$. By Lemma~\ref{lemma:real_part}, this guarantees that the maximum singular value of~\eqref{eq:y_tf} is finite for all $s \in \overline{\C}_+ \setminus 0$.
    
    To show the AC frequency and DC voltage of the system are $\mathscr{H}_{\infty}$ stable, we also demonstrate that $\sigma_{\text{max}}\left(H(s)\right)$ is bounded for $s \to 0$. Let there exist some $\delta > 0$ such that for all $s \in S_{\delta}$, $|g_i^{-1}(s)|$ is bounded for $i \in \mathcal{N}$, $\frac{1}{n}\sum_{i \in \mathcal{N}} \xi_i(s) > 0$, $\Re(\overline{k_j^{-1}}(s)) \geq d \in \R_{> 0}$ for all $j \in \left\{1,\,\dots,\,n_{\text{net}}^{\text{dc}}\right\}$, and $\Delta(s)$ converges linearly to zero as $s \to 0$. Then, using methods similar to~\cite[Thm. 3]{min_frequency_2025}, it can be shown that there exists some positive $\delta^\prime < \delta$ such that for all $s \in S_{\delta^\prime}$, $\lim_{s \to 0} \| H(s) - \bar{H}(s)\| = 0$ for some stable function $\bar{H}(s)$. Therefore, by Conditions~\ref{cond:acdc_universal}.\ref{acdc_cond:stable_ss_g}, \ref{cond:acdc_universal}.\ref{acdc_cond:avg_stab}, \ref{cond:acdc_dc_network}.\ref{acdc_cond:avg_k_spr}, and \ref{cond:acdc_dc_network}.\ref{acdc_cond:bounding_k_diff}, the system will converge to a stable steady state. 
    \qed{}

\bibliographystyle{IEEEtran}
\bibliography{IEEEabrv,ref}

\end{document}